\documentclass[a4paper,USenglish,cleveref]{lipics-v2021}

\pdfoutput=1
\hideLIPIcs
\nolinenumbers
\title{A Sheaf-Theoretic Characterization of Tasks in Distributed Systems}

\author{Stephan Felber}{TU Wien, Vienna, Austria}{stephan.felber@ecs.tuwien.ac.at}{0009-0003-6576-1468}{}
\author{Bernardo Hummes Flores}{École Polytechnique, Palaiseau, France}{bernardo.hummes-flores@lix.polytechnique.fr}{0000-0003-2325-1497}{}
\author{Hugo Rincon Galeana}{TU Berlin, Berlin, Germany}{hugo.galeana@ecs.tuwien.ac.at}{0000-0002-8152-1275}{}

\authorrunning{S. Felber, B. Hummes Flores, H. Rincon-Galeana}



\usepackage{amsmath,amsfonts,amssymb,bm}
\usepackage{graphicx}
\usepackage{mathtools}
\usepackage{float}
\usepackage{xspace}
\usepackage[dvipsnames]{xcolor}
\usepackage{environ}
\usepackage{tikzit}

\tikzstyle{LLa}=[fill=white, draw=black, shape=rectangle, inner sep=3pt, rounded corners]
\tikzstyle{LLb}=[fill={rgb,255: red,64; green,64; blue,64}, draw=black, shape=circle]
\tikzstyle{white circle}=[fill=white, draw=none, shape=circle, inner sep=1pt]
\tikzstyle{stalk}=[fill=white, draw=none, shape=rectangle]
\tikzstyle{terminal}=[fill=white, draw=red, shape=rectangle, rounded corners, inner sep=3pt, align=center]
\tikzstyle{edge}=[fill=white, draw=black, shape=circle]
\tikzstyle{textnode}=[fill={rgb,255: red,236; green,236; blue,236}, draw=none, shape=rectangle, align=left, rounded corners]
\tikzstyle{partially terminated}=[fill={rgb,255: red,255; green,213; blue,0}, draw=black, shape=rectangle, rounded corners, inner sep=3pt]
\tikzstyle{stalk 45}=[fill=none, draw=none, shape=circle, rotate=45, inner sep=1pt]
\tikzstyle{conflict}=[fill=white, shape=rectangle, draw=red, rounded corners, inner sep=4pt, align=center, line width=2pt]
\tikzstyle{injective arrow}=[->>, draw=black]
\tikzstyle{process a}=[fill=white, draw={rgb,255: red,240; green,72; blue,6}, shape=circle]
\tikzstyle{process b}=[fill=white, draw=blue, shape=circle]

\tikzstyle{dashed grey arrow}=[->, dotted, line width=0.8pt, draw={rgb,255: red,191; green,191; blue,191}, fill=none]
\tikzstyle{dashed arrow}=[->, dotted, line width=1pt]
\tikzstyle{a line}=[-, draw={rgb,255: red,240; green,72; blue,6}, line width=1pt]
\tikzstyle{b line}=[-, draw=blue, line width=1pt]
\tikzstyle{dotted b line}=[-, draw=blue, dashed]
\tikzstyle{dotted a line}=[-, draw={rgb,255: red,240; green,72; blue,6}, dashed]
\tikzstyle{restriction}=[->]
\tikzstyle{no restriction}=[draw={rgb,255: red,255; green,213; blue,0}, -, line width=1pt]
\tikzstyle{comm arrow}=[->]
\tikzstyle{dotted line}=[-, dotted, line width=1pt]
\tikzstyle{dashed undirected}=[-, dashed, draw=black]
\tikzstyle{causal consistency}=[-, line width=1pt, dashed, draw=red]
\tikzstyle{causal consistency b}=[-, draw=blue, line width=1pt, dashed]
\tikzstyle{mono}=[->, {-{Hooks[left]}}]
\tikzstyle{arrow red}=[draw=red, ->]
\tikzstyle{message arrival}=[->, line width=1pt]

\tikzstyle{red dot}=[fill=red, draw=black, shape=circle]
\tikzstyle{green dot}=[fill={rgb,255: red,0; green,128; blue,128}, draw=black, shape=circle]
\tikzstyle{new style 0}=[fill=blue, draw=black, shape=circle]
\tikzstyle{bi dash red}=[draw=red, dashed]
\tikzstyle{bi dash blue}=[draw=blue, dashed]
\tikzstyle{circ red}=[fill=none, draw=red, shape=circle]
\tikzstyle{circ blue}=[fill=none, draw=blue, shape=circle]
\tikzstyle{circ green}=[fill=none, draw={rgb,255: red,0; green,128; blue,128}, shape=circle]
\tikzstyle{circ black}=[fill=none, draw=black, shape=circle]

\tikzstyle{uni}=[draw=black, ->, fill=none]
\tikzstyle{bi}=[fill={rgb,255: red,128; green,128; blue,128}, draw=black, <->]
\tikzstyle{uni dash}=[fill=none, ->, dashed]
\tikzstyle{bi dash}=[-, dashed]
\tikzstyle{uni dash red}=[dashed, draw=red, ->]
\tikzstyle{uni dash blue}=[dashed, draw=blue, ->]
\tikzstyle{bi dash red}=[-, draw=red, dashed]
\tikzstyle{uni blue}=[draw=blue, ->]
\tikzstyle{uni red}=[draw=red, ->]
\tikzstyle{bi blue}=[-, draw=blue, <->]
\tikzstyle{bi red}=[draw=red, <->]
\tikzstyle{bi dash blue}=[-, draw=blue, dashed]
\tikzstyle{uni green dash}=[draw={rgb,255: red,0; green,128; blue,128}, ->, dashed]
\tikzstyle{uni green}=[draw={rgb,255: red,0; green,128; blue,128}, ->]

\newcommand{\demph}[1]{\textbf{#1}} 
\newcommand{\remph}[1]{\emph{#1}}   

\newcommand{\Cat}[1]{\mathbf{#1}}   
\newcommand{\To}{\rightarrow}       

\renewcommand{\leq}{\leqslant}
\renewcommand{\epsilon}{\varepsilon}

\newcommand{\angles}[1]{\left\langle #1 \right\rangle}

\newcommand{\N}{\mbox{$\mathbb{N}$}}
\newcommand{\Z}{\mbox{$\mathbb{Z}$}}

\newcommand{\GblState}[1]{\angles{#1}} 
\newcommand{\ids}{{\ensuremath{[n]}}}
\newcommand{\csim}{{\ensuremath{\simeq}}}
\newcommand{\F}{\ensuremath{\mathcal{F}}}

\newcommand{\g}{\ensuremath{\GblState{g}}}
\newcommand{\h}{\ensuremath{\GblState{h}}}
\newcommand{\w}{\ensuremath{\GblState{w}}}
\newcommand{\ccl}{\ensuremath{ccl}}
\newcommand{\exec}{\ensuremath{\mathcal{E}}}

\newcommand{\eframe}{\ensuremath{\mathcal{K}}}
\newcommand{\slice}{\ensuremath{\mathcal{N}}}
\newcommand{\systemslice}{\ensuremath{\mathcal{S}}}
\newcommand{\gstates}{\ensuremath{\mathrm{C}}}
\newcommand{\gstate}{\ensuremath{\mathrm{c}}}
\newcommand{\run}{\ensuremath{\sigma}}
\newcommand{\runs}{\ensuremath{\Sigma}}
\newcommand{\proc}{\ensuremath{p}}
\newcommand{\processes}{\ensuremath{\Pi}}

\newcommand{\protocol}{\ensuremath{\mathcal{P}}}
\newcommand{\task}{\ensuremath{T}}

\newcommand{\fdeci}{\ensuremath{\delta}}

\begin{document}
\maketitle
\begin{abstract}
We introduce a sheaf-theoretic characterization of task solvability in general distributed computing models, unifying distinct approaches to message-passing models. We establish cellular sheaves as a natural mathematical framework for analyzing the global consistency requirements of local computations. Our main contribution is a task sheaf construction that explicitly relates both the distributed system and task, in which terminating solutions are precisely its global sections. We prove that a task can only be solved by a system when such sections exist in a task sheaf obtained from an execution cut, the frontier in which processes have enough information to decide. Our characterization is model-independent, working under varying synchronicity, failures and message adversaries, as long as the model produces runs composed of global states of the system. Furthermore, we show that the cohomology of the task sheaf provides a linear algebraic description of the decision space of the processes, and encodes the obstructions to find solutions. This opens way to a computational approach to protocol synthesis, which we illustrated by deriving a protocol for approximate agreement. This work bridges distributed computing and sheaf theory, providing both theoretical foundations for analyzing task solvability and tools for protocol design leveraging computational topology.
\end{abstract}

\section{Introduction} \label{sec:intro} 

Assessing the solvability or impossibility of a given task in a distributed system is at the center of distributed systems research. A variety of theoretical tools have been developed by sourcing from different fields such as combinatorics, topology and epistemic logic. One of the earliest results is the two generals problem impossibility \cite{AEH75,lamport1983weak}, which roughly states that consensus is impossible in networks with message loss and unbounded delays. Numerous results have followed since: the celebrated FLP \cite{FLP85} impossibility shows that consensus cannot be achieved in asynchronous and crash-prone systems, the lossy-link impossibility \cite{SW89} (the authors referred to it as \emph{mobile failures}) that shows that consensus is impossible even in a synchronous two-process model if just one message may be lost in each round, \cite{HM90} characterized consensus in distributed systems via \remph{epistemic logic} and \cite{NSW24} provides a \remph{topological} characterization of consensus under general message adversaries.

Although the aforementioned results are mostly related to problem solvability, they are specifically focused on consensus. Nevertheless, these results undoubtedly contributed to the development of frameworks that would later address general task solvability. Some such noteworthy results are  Herlihy and Shavit's characterization of wait-free t-resilient shared memory distributed computation, via combinatorial topology \cite{HS93}, Attiya, Nowak and Castañeda's characterization of terminating tasks \cite{AttiyaCastanedaNowak23}, and Alcántara, Castañeda, Flores and Rajsbaum's characterization of look-compute-move wait-free robot tasks \cite{ACFR19}.

Sheaf theory originated from studying how local constraints give rise to global solutions, paired with cohomology theory as a measure of obstructions to such global constructions. Cellular sheaves were introduced by Curry \cite{Curry14} as a combinatorial counterpart to the common notion of sheaves found in the literature \cite{MacLaneMoerdijk94}, with the interest of having a theory that is computable. Most importantly to us, the theory of cellular sheaves presents the concept of sheaves in an approachable manner, whence it becomes clear its usefulness to depict global properties of local information, as is the case for distributed computing. Sheaves have proven valuable in studying similar local-to-global phenomena, from opinion dynamics \cite{HansenGhrist21} to contextuality in quantum mechanics \cite{AbramskyBrandenburger11} and sensor integration \cite{Robinson17,JoslynCharlesDePernoGouldNowakPraggastisPurvineRobinsonStrulesWhitney20}. We now extend this framework to distributed systems.

In this paper, we extend the brief announcement \cite{felberhummesfloresrincongaleana25} introducing a generic sheaf-theoretic framework for analyzing task solvability parametrized by indistinguishability relations on global states of the system, forming a graph-like structure. These indistinguishability graphs are closely related to Kripke frames in epistemic logic, as explored in \cite{bookof4}. Furthermore, our framework enables us to leverage the computational methods in cohomology theory so to analyze task solvability, and even derive solving algorithms in systems where the indistinguishability graph admits a cohomology group construction.

\textbf{Contributions} 
We introduce a novel sheaf-theoretic framework for analyzing distributed task solvability. Our approach complements the existing frameworks by connecting category theory with distributed computing. This novel bridge allows us to directly apply computational methods for decision function synthesis using linear algebra. We expect this first connection to provide fertile ground for more applications in distributed computing.

More precisely, we provide the following.
\begin{itemize}
  \item A task sheaf construction (\Cref{def:task-sheaf}) that captures both the local knowledge of each process (local constraints) and relates it to valid solutions, assessing global consistency;
  \item A finite substructure of the executions that preserves the task solvability as system slices (\Cref{def:system-slice});
  \item An equivalence between the existence of solutions to a task and the existence of sections over the corresponding system slice (\Cref{thm:term_task_solvable});
  \item A way to use cohomology to transform impossibility arguments into computational problems (\Cref{thm:computable_delta}).
\end{itemize}

\textbf{Paper Organization} 
We start by defining the model of a distributed computing system we use in \Cref{sec:model}, giving special attention to the discrete structures used to depict the system executions (\Cref{sec:frames}) and its tasks (\Cref{sec:tasks}). \Cref{sec:system-slices} introduces system slices, the smallest sufficient structure needed to assess task solvability, and how it brought to the categorical machinery. \Cref{sec:sheaf-model} defines the task sheaf, first by its explicit construction (\Cref{sec:task-sheaf-explicit}), then by its equivalent categorical formulation (\Cref{sec:task-sheaf-categorical}), so to state the main theorem relating the solvability of tasks to the existence of sections (\Cref{sec:task-sheaf-solvability}). \Cref{sec:cohomology} uses the categorical machinery of task sheaves to define its cohomology construction, so used to computationally assess task solvability. \Cref{sec:conc} provides some concluding words.

\section{System Model} 
\label{sec:model}
We now define how processes evolve, communicate and make decisions, in order to define the runs of that system. Our main contribution on task solvability \emph{solely depends} on the runs and the task specification and consequently our approach is applicable to any system providing runs consisting of global states consisting of local states. As there are multiple approaches to modeling distributed systems, we do not formally specify how to obtain said runs, but exemplify it via a synchronous message passing system, which is one common choice to represent a distributed system (e.g., \cite{bookraynal13}).

We call \emph{local protocols} processes \(\processes = \{\proc\}_{i \in \ids}\) which are deterministic state machines \(\protocol_p\) containing the possible \demph{local states}, together with distinguished input states. The local protocol is defined by \demph{transition function}, a \demph{communication function} and a \demph{decision function} \(\fdeci_p \) mapping local states to task-specific outputs. Processes evolve depending on their previous local state and some partial information about other processes local states, usually represented via messages they received. The \demph{possibilities} of (partial) information reception of the processes about the whole system state is captured as the \demph{adversary}, also called the \demph{message adversary} \cite{SW89,santoro07,schmid_strongest_2018}, \emph{heard-of predicates} \cite{charron-bost09}, \demph{time-varying-graphs} \cite{casteigts_time-varying_2012}, \demph{scheduler} \cite{moses_layered_2002} etc., which influence the behavior of the system. Since there is a pallette of different adversaries, each different in their way they impact the system, and the precise capabilities of the adversary \emph{are not important in stating our results}, we do not formalize it further. The only object we need is a \demph{run} consisting of a sequence of \demph{configurations}, defined in what follows.

\subsection{System Runs and System Frames} 
\label{sec:frames}

A \demph{global configuration} \(\g\) is a tuple of local states, one per process. A projection function \(\pi_{\proc}: S^n \to S_{\proc_i}\) is defined as mapping configurations to local states. A projection to the inputs $\pi_I: S^n \to I^n$ is defined as retrieving the vector of input states, and accordingly they can be composed $\pi_{I_{i}} = \pi_{\proc_i} \circ \pi_I : S^n \to I_{\proc}$, projecting to $p_j$'s input. A configuration is called \demph{terminal} if all processes' decision functions map to a valid output value. If this is true for some processes, but not all, then the configuration is called \demph{partially terminated}.
A decision function \(\delta\) that always eventually provides a terminal configuration is called \demph{terminating}. Additionally, we assume that the input values of a process are encoded into their local states and are not forgotten. A sequence of global configurations is an \demph{execution} or a \demph{run} of the system, we identify the \demph{system} with the set of all of its runs, denoted $\Sigma$.

\begin{example}[Synchronous message-passing system with lossy link adversary]
    For example, the \emph{synchronous lossy link} model for two processes $\Pi = \{ a,b \}$ produces a system frame. Informally a \demph{synchronous} message adversary evolves in synchronous rounds where processes compute a new state from the messages received in the previous round and the current state \demph{simultaneously}. Messages sent at the beginning of one round either arrive at the end of that round, or are lost forever. The \emph{message adversary} is allowed to drop at most one message in each round, implying that processes know that if they receive no message at the end of a round, then their own message arrived. Written differently, possible message arrival in each round may be indicated by an edge in $\{ \leftarrow, \leftrightarrow, \rightarrow \}$, i.e., the no-arrow case is excluded. We assume that processes remember everything, i.e., they have unlimited memory to keep all their history and additionally, they transmit their whole history in every round (in the literature this is know as the \emph{full information protocol} and mostly used for impossibility arguments, see for example \cite{FichRuppert03} where they call it a \emph{full information algorithm}).
    
    It is well known that terminating consensus is already impossible here, for example Santoro and Widmayer proved it in \cite{SW89}.

    \begin{figure}[ht]
        \centering
        \makebox[\textwidth][c]{\begin{tikzpicture}
	\begin{pgfonlayer}{nodelayer}
		\node [style=LLa] (18) at (-1, -1) {$1,0$};
		\node [style=LLa] (28) at (4, -1) {$1,0\rightarrow$};
		\node [style=LLa] (29) at (9, -1) {$1,0\rightarrow\leftrightarrow$};
		\node [style=LLa] (30) at (14, -1) {$1,0\rightarrow\leftrightarrow\leftarrow$};
		\node [style=LLa] (37) at (19, -1) {$1,0\rightarrow\leftrightarrow\leftarrow\leftarrow$};
		\node [style=process a] (38) at (-1, 4) {a};
		\node [style=process b] (39) at (-1, 1) {b};
		\node [style=process a] (40) at (4, 4) {a};
		\node [style=process b] (41) at (4, 1) {b};
		\node [style=process a] (42) at (9, 4) {a};
		\node [style=process b] (43) at (9, 1) {b};
		\node [style=process a] (44) at (14, 4) {a};
		\node [style=process b] (45) at (14, 1) {b};
		\node [style=process a] (46) at (19, 4) {a};
		\node [style=process b] (47) at (19, 1) {b};
		\node [style=none] (48) at (1.5, 6) {1:};
		\node [style=none] (49) at (6.5, 6) {2:};
		\node [style=none] (50) at (11.5, 6) {3:};
		\node [style=none] (51) at (16.5, 6) {4:};
		\node [style=none] (53) at (-5, -1) {Global States:};
		\node [style=none] (54) at (-5, 6) {Rounds:};
		\node [style=none] (55) at (-5, 2.5) {Messages:};
		\node [style=none] (56) at (0.75, 2.5) {};
		\node [style=none] (57) at (2.25, 2.5) {};
		\node [style=none] (58) at (5.75, 2.5) {};
		\node [style=none] (59) at (7.25, 2.5) {};
		\node [style=none] (60) at (10.75, 2.5) {};
		\node [style=none] (61) at (12.25, 2.5) {};
		\node [style=none] (62) at (15.75, 2.5) {};
		\node [style=none] (63) at (17.25, 2.5) {};
	\end{pgfonlayer}
	\begin{pgfonlayer}{edgelayer}
		\draw [style=message arrival] (43) to (42);
		\draw [style=message arrival] (47) to (46);
		\draw [style=message arrival] (42) to (43);
		\draw [style=message arrival] (40) to (41);
		\draw [style=message arrival] (45) to (44);
		\draw [style=dashed arrow] (18) to (28);
		\draw [style=dashed arrow] (28) to (29);
		\draw [style=dashed arrow] (29) to (30);
		\draw [style=dashed arrow] (30) to (37);
		\draw [style=dashed arrow] (56.center) to (57.center);
		\draw [style=dashed arrow] (58.center) to (59.center);
		\draw [style=dashed arrow] (60.center) to (61.center);
		\draw [style=dashed arrow] (62.center) to (63.center);
	\end{pgfonlayer}
\end{tikzpicture}}
        \caption{The first line counts the number of rounds. In the middle we depict the actual message reception, where an arrow from $a$ to $b$ indicates that $a$'s message to $b$ arrived in that round. At the bottom we trace the resulting \emph{global states} in that specific run. The first global state only contains the input values:$1$ for $a$ and $0$ for $b$. The second global state records the first and which messages arrived in the first round. The third global state records the second and which messages arrived in the second round and so on ...}
        \label{fig:placeholder}
    \end{figure}
\end{example}

\label{sec:system_frame}
We can now define execution graphs that do not depend on the system's decision functions.

\begin{definition}[Execution Graph]\label{def:execution-graph}
  Let $\runs$ be a set of runs. The \demph{execution graph} of \(\runs\) is a directed graph $\exec_{\runs} = (V(\exec_\runs), E(\exec_\runs))$ where:
  \begin{align*}
    V(\exec_\runs) &:= \{ \g \mid \exists \run \in \runs, \exists i \in \mathbb{N} \text{ such that } \run_i =\g \},\\
    E(\exec_\runs) &:= \{ (\g, \h) \mid \exists \run \in \runs, \exists i \in \mathbb{N} \text{ such that } \run_i=\g, \, \run_{i+1}=\h \}.
  \end{align*}
\noindent The uniquely generated execution graph generated by the system $\runs$ is denoted \(\exec_{\runs}\).
\end{definition}

From the system model established in \Cref{sec:model}, we formalize the requirement that the decision function must be deterministic. This takes the form of a system frame, where any set of configurations that constitute a run are equipped with process-wise indistinguishability relations. The system frame will later parametrize where, in the space of configurations, a decision function must behave constantly for lack of distinguishing power between identical states of the same process.

\begin{definition}[Configuration Indistinguishability] \label{def:sim-cfg}
  Two configurations $\g$, $\h$ are \demph{indistinguishable} for process $\proc$, denoted $\g \sim_{\proc} \h$, iff $\proc$ has the same local state in both $\g$ and $\h$, i.e., $\pi_{\proc}(\g) = \pi_{\proc}(\h)$.
\end{definition}

Edges on the execution graph are also called \demph{causal links}, as they arrange the configurations in a causal sequences that follows the execution of a possible run. An acyclic execution graph also induces a partial order $\leq$ over configurations, we write $\g \leq \h$ if there exists a path from $\g$ to $\h$.

We define now the system frame, which formalizes the information available to each process throughout the possible executions of the distributed system. This structure will be later used in \Cref{def:task-sheaf}, on the construction of the task sheaf, to make precise the idea that a process must always choose the same value when the information available to it is the same.

\begin{definition}[System Frame]\label{def:sysframe}
  Let $\exec_{\runs}$ be the execution graph generated by the set of runs \(\runs\), and $\{\sim_{\proc}\}_{\proc\in\processes}$ the set of indistinguishability relations on $V(\exec)$, indexed by each process $\proc \in \Pi$. We say that $\eframe := (V(\exec), \{\sim_{p}\}_{p \in \Pi})$ is the \demph{system frame} induced by $\runs$.
\end{definition}

\begin{remark}[Equivalence Relation on the System Frame]\label{rmk:equiv-frame}
  Note the indistinguishability relation on local states defines an equivalence relation.
\end{remark}

See in \Cref{ex:system-frame} for the explicit construction of a system frame after one step.

\begin{example} \label{ex:system-frame}
  A system frame after zero and one step corresponding in the lossy-link synchronous message adversary from the previous example is shown in \Cref{fig:LL-system-frame}. The frame after zero steps could be written as $V(\exec_{LL}) = \{(0,0), (0,1), (1,0), (1,1)\}$ and $\sim_a = \{ ((0,0),(0,1)), ((1,0),(1,1)) \}$, $\sim_b = \{ ((0,0),(1,0)), ((0,1),(1,1)) \})$.
  This is the finest execution graph possible under the lossy-link adversary. Dotted arrows are the causal links. The colored edges denote the indistinguishability edges between configurations, orange for $a$, green for $b$. The nodes here represent the configurations: for example, after the first step, $a$ cannot distinguish between $(0,0\rightarrow)$ from $(0,1\rightarrow)$ as $a$ hasn't received a message in neither configuration, whereas $b$ can distinguish the two configurations because it has a different inputs in them.

  \begin{figure}[ht]
    \makebox[\textwidth][c]{\begin{tikzpicture}
	\begin{pgfonlayer}{nodelayer}
		\node [style=LLa] (5) at (3, -6) {$1,1$};
		\node [style=LLa] (6) at (3, 0) {$0,0$};
		\node [style=LLa] (17) at (3, -9) {$0,1$};
		\node [style=LLa] (18) at (3, -3) {$1,0$};
		\node [style=LLa] (19) at (13, 0) {$0,0\leftrightarrow$};
		\node [style=LLa] (20) at (9, 0) {$0,0\rightarrow$};
		\node [style=LLa] (21) at (17, 0) {$0,0\leftarrow$};
		\node [style=LLa] (22) at (9, -9) {$0,1\rightarrow$};
		\node [style=LLa] (23) at (13, -9) {$0,1\leftrightarrow$};
		\node [style=LLa] (24) at (17, -9) {$0,1\leftarrow$};
		\node [style=LLa] (25) at (17, -6) {$1,1\leftarrow$};
		\node [style=LLa] (26) at (13, -6) {$1,1\leftrightarrow$};
		\node [style=LLa] (27) at (9, -6) {$1,1\rightarrow$};
		\node [style=LLa] (28) at (9, -3) {$1,0\rightarrow$};
		\node [style=LLa] (29) at (13, -3) {$1,0\leftrightarrow$};
		\node [style=LLa] (30) at (17, -3) {$1,0\leftarrow$};
		\node [style=none] (53) at (22, -4.5) {};
		\node [style=none] (58) at (20, -4.5) {};
	\end{pgfonlayer}
	\begin{pgfonlayer}{edgelayer}
		\draw [style=dashed arrow, bend left=15] (6) to (20);
		\draw [style=dashed arrow, bend left=15] (6) to (19);
		\draw [style=dashed arrow, bend left=15] (6) to (21);
		\draw [style=a line, in=180, out=-180, looseness=0.75] (6) to (17);
		\draw [style=a line] (5) to (18);
		\draw [style=b line] (19) to (20);
		\draw [style=b line] (22) to (23);
		\draw [style=b line, bend left=90, looseness=1.25] (25) to (24);
		\draw [style=b line] (26) to (27);
		\draw [style=b line] (28) to (29);
		\draw [style=b line, bend right=90, looseness=1.25] (30) to (21);
		\draw [style=b line] (18) to (6);
		\draw [style=b line] (17) to (5);
		\draw [style=a line] (21) to (19);
		\draw [style=a line, bend right=90, looseness=0.75] (20) to (22);
		\draw [style=a line] (23) to (24);
		\draw [style=a line] (25) to (26);
		\draw [style=a line] (27) to (28);
		\draw [style=a line] (29) to (30);
		\draw [style=dashed arrow, bend left=15] (18) to (28);
		\draw [style=dashed arrow, bend left=15] (18) to (29);
		\draw [style=dashed arrow, bend left=15] (18) to (30);
		\draw [style=dashed arrow, bend right=15] (5) to (27);
		\draw [style=dashed arrow, bend right=15] (5) to (26);
		\draw [style=dashed arrow, bend right=15] (5) to (25);
		\draw [style=dashed arrow, bend right=15] (17) to (22);
		\draw [style=dashed arrow, bend right=15] (17) to (23);
		\draw [style=dashed arrow, bend right=15] (17) to (24);
		\draw [style=dashed arrow] (58.center) to (53.center);
	\end{pgfonlayer}
\end{tikzpicture}}
    \caption{Depiction of the system frame up to step $1$ of the lossy-link synchronous message adversary together with a full-information protocol. This system frame grows to the right.}
    \label{fig:LL-system-frame}
  \end{figure}

  An example for the causality relation would be the configuration $(0,0\leftrightarrow)$ depending on the configuration $(0,0)$, written as $(0,0)\leq(0,0\leftrightarrow)$.
\end{example}

\subsection{Distributed Tasks and their Algebraic Structure} 
\label{sec:tasks}
Given a system, we can now talk about the distributed tasks for which we will provide a sheaf-theoretic perspective.

\begin{definition}[Tasks] \label{def:task}
  A \demph{task} is a triple \(\task = \langle \mathcal{I}, \mathcal{O}, \Delta \rangle\), where \(\mathcal{I}\) is the set of possible input vectors, \(\mathcal{O}\) is the set of possible output vectors and \(\Delta: \mathcal{I} \to 2^\mathcal{O}\) is a map associating to each input vector the set of valid output vectors. \(\mathcal{I}_\proc\), \(\mathcal{O}_\proc\) denote the possible inputs and outputs restricted to a process \(\proc\).
\end{definition}

\begin{definition}[Terminating Task Solvability] \label{def:task-validity}
  A decision function \(\delta_i\) is said to \demph{solve} a task \(\task\) if for every run \(\run = (\g)_{t \in \N} \in \runs\) the following holds:
  \begin{itemize}
    \item \textbf{Termination}: For every process \(\proc_j \in Pi\), there is a step \(t \in \N\) such that \(\fdeci(\pi_j(\g_t)) \neq \bot\),
    \item \textbf{Validity}: There is an output vector \(o \in \Delta(\pi_I(\g))\), such that for every process \(\proc_j \in Pi\) there is a step \(t \in \N\) where \(\fdeci(\pi_j(\g_t)) = o_\proc\) and for every \(t' \le t\), \(\fdeci(\pi_j(\g_{t'})) = \bot\).
    \end{itemize}
\end{definition}

Termination requires every process to eventually decide on some value, and validity requires that all individual decisions correspond to a valid output configuration for the respective input configuration. Note that the decision of process \(\proc\) corresponds to the first value obtained by its decision function other than \(\bot\).

The topological approach to distributed computing \cite{HerlihyKozlovRajsbaum14} provides a combinatorial structure for defining tasks as simplicial complexes. Recall that a \demph{chromatic simplicial complex} is a set of vertices and a set of simplices defined from the vertices: a subset of the powerset of vertices that is closed under inclusion. Each vertex has a color, the process identity, and each face can depict at most one instance of each color. Vertices are local states and faces are then global states.

The sets of input \(\mathcal{I}\) and output \(\mathcal{O}\) values are now chromatic simplicial complexes, and the task specification \(\Delta : \mathcal{I} \to 2^{\mathcal{O}}\) works the same, and a task consists of a relation between valid outputs simplices for a given set of inputs simplices.
This formulation captures the combinatorial nature of a task specification. It which will be later used in \Cref{sec:sheaf-model} as the data of the task that must be tracked by the processes.

\section{From System Frames to Categories} 
\label{sec:system-slices}

Termination on an execution graph can be represented as a set of configurations that \demph{cuts} the execution graph in half. Indeed, termination by definition is just a set of configurations that intersects any run, together with some extra conditions that we will define in the following. Unfortunately, not every such cut through the execution graph need be finite and therefore easy to find. Thus we do not concern ourselves with how we look for one, instead describe its shape.

\begin{definition}[Execution Cut]\label{def:excut}
    Let $A \subset V(\exec)$ be a set of configurations. We say that $A$ is an \demph{execution cut} iff it is a cut set in $\exec$, i.e., it intersects every run in $\exec$.
\end{definition}

An execution cut represents a set of ``unavoidable'' configurations within the system. Furthermore, we say that an execution cut $A \subset V(\exec)$ is \demph{terminal} iff any $U \in A$ is terminal, i.e., at configuration $U$, all processes must have decided.

\begin{definition}[Local Star]\label{def:localstar}
    Let $\g \in V(\exec)$ be a configuration, $\proc \in \Pi$, and $N_{\proc}(\g):= \{\h \in V(\exec) \mid \g \sim_{\proc} \h \}$. That is, $N_\proc(\g)$ is the equivalence class of $\g$ under $\sim_\proc$. We define the \demph{local star} of $\g$, denoted by $\slice(\g)$, as the labeled graph obtained by $\bigcup_{\proc \in \Pi} N_\proc(\g)$ and extend it over sets by $\slice(A) = \bigcup_{\g \in A} \slice(\g)$
\end{definition}

\begin{remark}
  Note that for any $\proc \in \Pi$, $N_\proc(\g)$ induces a complete graph with all edges labeled by $\proc$.
\end{remark}

\begin{definition}[Causal Closure]\label{def:causalclosure}
    Let $A$ be a terminal execution cut and $\slice(A)$ its local star. The \remph{causal closure} $\ccl()$ contains all configurations that lie between $\slice(A)$ and $A$:
    \[\ccl(\slice(A)) = \{ \g \in A \mid \exists \h \in \slice(A), \exists \w \in A, \h \leq \g \leq \w \}\]
\end{definition}

A \demph{causal closure} $\ccl(\slice(A))$ of a terminal execution cut $A$ extends to the partially terminated configurations in $\slice(A) \setminus A$, where at least one process $p$ has already terminated, but not necessarily all of them. Any partially terminated configuration \g eventually results into a terminated configuration \h in $A$. The causal closure contains all successors of partially terminated configurations up until they result in a fully terminated configuration in $A$.
As some processes have already decided in a partially terminated configuration \g, we ensure that they keep their decided values in any successor configurations \h (reflecting that decisions are final) by constructing the \demph{causal consistency} $\{\simeq_{\proc}\}_{\proc\in\processes}$ relation such that $\g\simeq_{p}\h$.

\begin{definition}[Causal Consistency Relation] \label{def:consistency_ensuring}
    We define the \demph{causal consistency} relation $\{ \csim_{\proc} \}_{\proc\in\processes}$ as the symmetric closure of the binary relation composition $\leq\circ\sim_{\proc}$.
\end{definition}

Intuitively, given a causal closure of a terminal execution cut $\ccl(\slice(A))$, the causal consistency relation relates all configurations where a process terminated in a \demph{preceding} or \demph{indistinguishable} configuration. As $\leq$ and $\{\sim_{\proc}\}_{\proc\in\processes}$ are reflexive, $\{\simeq_{\proc}\}_{\proc\in\processes}$ is reflexive, and contains both relations $\leq$ and $\{\sim_{\proc}\}_{\proc\in\processes}$. Observe that any causally dependent configurations in $\g,\h\in A$ such that $\g \leq \h$ are related $\g\simeq\h$ (and also $\h\simeq\g$) for all processes. Configurations \g in $\ccl(\slice(A))$ that are not in $A$, are related to a configuration \h in $A$ via the indistinguishability relation of $p$ and any successor of \g is also related to \h for $p$. This ensures that the partially terminated process $p$ keeps its terminated value.

\begin{definition}[System Slice] \label{def:system-slice}
  Let $A \subset V(\exec)$ be an execution cut in the system frame $\eframe$, we define a \demph{system slice}, denoted by $\systemslice_A = (\gstates, \{ \csim_{\proc} \}_{\proc\in\processes})$, as a tuple consisting of the causal closure over the local star over $A$, i.e., \(\gstates = \ccl(\slice(A))\), together with its \demph{causal consistency relation}.
\end{definition}

In general system slices need not be finite, see \Cref{ex:tilted-consensus} for an explicit construction of an execution graph with a necessarily infinite system slice. An finite system slice is depicted in \Cref{ex:term-task-sheaf}

\begin{example}[System slice of single shot message adversary execution]\label{ex:tilted-consensus}
  In this example, we consider the \demph{tilted consensus task}, where both processes have to decide on $a$'s value. The synchronous communication adversary here allows exactly one message from $a$ to $b$ per run, the set of all runs therefore consists of all $\sigma_k = - \rightarrow^k -^\omega$, where $k<\infty$.
  In \Cref{fig:tilted-consensus} $a$'s indistinguishability relations between configurations are again orange, $b$'s are blue. The configurations within dashed boxes are terminated, the yellow boxes are partially terminated. The dashed red and blue edges represent the causal consistency relation (most transitive edges are omitted favoring readability).
  \begin{figure}[ht]
    \makebox[\textwidth][c]{\scalebox{0.9}{\begin{tikzpicture}
	\begin{pgfonlayer}{nodelayer}
		\node [style=LLa] (5) at (3, -6) {$1,1$};
		\node [style=LLa] (17) at (-3, -6) {$0,1$};
		\node [style=partially terminated] (6) at (-3, -3) {$0,1-$};
		\node [style=partially terminated] (18) at (3, -3) {$1,1-$};
		\node [style=LLa] (70) at (-8, -3) {$0,1\rightarrow$};
		\node [style=LLa] (71) at (8, -3) {$1,1\rightarrow$};
		\node [style=partially terminated] (72) at (-3, 0) {$0,1--$};
		\node [style=partially terminated] (73) at (3, 0) {$1,1--$};
		\node [style=LLa] (74) at (-8, 0) {$0,1-\rightarrow$};
		\node [style=partially terminated] (76) at (-3, 3) {$0,1---$};
		\node [style=partially terminated] (77) at (3, 3) {$1,1---$};
		\node [style=LLa] (78) at (-8, 3) {$0,1--\rightarrow$};
		\node [style=LLa] (83) at (8, 0) {$1,1-\rightarrow$};
		\node [style=LLa] (84) at (8, 3) {$1,1--\rightarrow$};
		\node [style=none] (89) at (-10.25, 4.25) {};
		\node [style=none] (95) at (-3, 5) {};
		\node [style=none] (96) at (3, 5) {};
		\node [style=none] (97) at (5.25, -7) {};
		\node [style=none] (98) at (5.25, 5) {};
		\node [style=none] (99) at (6.25, 6) {};
		\node [style=none] (100) at (6.25, -8) {};
		\node [style=none] (101) at (14, -8) {};
		\node [style=none] (102) at (15, -7) {};
		\node [style=none] (103) at (15, 5) {};
		\node [style=none] (104) at (14, 6) {};
		\node [style=none] (109) at (-6.25, -8) {};
		\node [style=none] (110) at (-5.25, -7) {};
		\node [style=none] (111) at (-5.25, 5) {};
		\node [style=none] (112) at (-6.25, 6) {};
		\node [style=none] (113) at (-14.5, 6) {};
		\node [style=none] (114) at (-15.5, 5) {};
		\node [style=none] (115) at (-15.5, -7) {};
		\node [style=none] (116) at (-14.5, -8) {};
		\node [style=textnode] (117) at (10.5, -6.5) {1-deciding terminal\\ configurations};
		\node [style=textnode] (118) at (-10.25, -6.5) {0-deciding terminal\\ configurations};
		\node [style=LLa] (120) at (-13, 0) {$0,1\rightarrow-$};
		\node [style=LLa] (121) at (-13, 3) {$0,1-\rightarrow-$};
		\node [style=LLa] (122) at (13, 0) {$1,1\rightarrow-$};
		\node [style=LLa] (123) at (13, 3) {$1,1-\rightarrow-$};
		\node [style=none] (124) at (10.25, 4.25) {};
		\node [style=none] (125) at (14.25, 0.75) {};
		\node [style=none] (126) at (14.25, 3.75) {};
		\node [style=none] (127) at (-14.25, 3.75) {};
		\node [style=none] (128) at (-14.25, 0.75) {};
	\end{pgfonlayer}
	\begin{pgfonlayer}{edgelayer}
		\draw [style=b line] (18) to (6);
		\draw [style=b line] (17) to (5);
		\draw [style=dashed arrow] (17) to (6);
		\draw [style=dashed arrow] (5) to (18);
		\draw [style=a line] (70) to (6);
		\draw [style=a line, in=180, out=0] (18) to (71);
		\draw [style=b line, in=360, out=180] (73) to (72);
		\draw [style=a line] (74) to (72);
		\draw [style=dashed arrow] (6) to (72);
		\draw [style=dashed arrow] (18) to (73);
		\draw [style=b line] (77) to (76);
		\draw [style=a line] (78) to (76);
		\draw [style=dashed arrow] (72) to (76);
		\draw [style=dashed arrow] (73) to (77);
		\draw [style=a line] (73) to (83);
		\draw [style=a line] (77) to (84);
		\draw [style=dotted line] (77) to (96.center);
		\draw [style=dotted line] (76) to (95.center);
		\draw [style=dotted line] (78) to (89.center);
		\draw [style=dashed undirected] (100.center) to (101.center);
		\draw [style=dashed undirected, bend right=45] (101.center) to (102.center);
		\draw [style=dashed undirected] (102.center) to (103.center);
		\draw [style=dashed undirected, bend right=45] (103.center) to (104.center);
		\draw [style=dashed undirected] (104.center) to (99.center);
		\draw [style=dashed undirected, bend left=315] (99.center) to (98.center);
		\draw [style=dashed undirected] (98.center) to (97.center);
		\draw [style=dashed undirected, bend right=45] (97.center) to (100.center);
		\draw [style=dashed undirected, bend right=45] (109.center) to (110.center);
		\draw [style=dashed undirected] (110.center) to (111.center);
		\draw [style=dashed undirected, bend right=45] (111.center) to (112.center);
		\draw [style=dashed undirected] (116.center) to (109.center);
		\draw [style=dashed undirected, bend right=315] (116.center) to (115.center);
		\draw [style=dashed undirected] (115.center) to (114.center);
		\draw [style=dashed undirected, bend left=45] (114.center) to (113.center);
		\draw [style=dashed undirected] (113.center) to (112.center);
		\draw [style=dashed arrow] (5) to (71);
		\draw [style=dashed arrow] (18) to (83);
		\draw [style=dashed arrow] (73) to (84);
		\draw [style=dashed arrow] (72) to (78);
		\draw [style=dashed arrow] (6) to (74);
		\draw [style=dashed arrow] (17) to (70);
		\draw [style=causal consistency, bend left=15] (18) to (73);
		\draw [style=causal consistency] (71) to (83);
		\draw [style=causal consistency] (83) to (84);
		\draw [style=causal consistency, bend left=15] (73) to (77);
		\draw [style=causal consistency, bend left=15] (71) to (18);
		\draw [style=causal consistency, bend left=15] (84) to (77);
		\draw [style=causal consistency] (70) to (74);
		\draw [style=causal consistency] (74) to (78);
		\draw [style=causal consistency, bend right=15] (70) to (6);
		\draw [style=causal consistency, bend right=15] (74) to (72);
		\draw [style=causal consistency, bend right=15] (78) to (76);
		\draw [style=causal consistency, bend left=15] (6) to (72);
		\draw [style=causal consistency, bend left=15] (72) to (76);
		\draw [style=dashed arrow] (70) to (120);
		\draw [style=dashed arrow] (74) to (121);
		\draw [style=causal consistency, bend left=15] (70) to (120);
		\draw [style=causal consistency b, bend right=15] (70) to (120);
		\draw [style=causal consistency b, bend right=15] (74) to (121);
		\draw [style=causal consistency, bend left=15] (74) to (121);
		\draw [style=dashed arrow] (71) to (122);
		\draw [style=causal consistency, bend right=15] (73) to (83);
		\draw [style=dashed arrow] (83) to (123);
		\draw [style=causal consistency, bend left=15] (71) to (122);
		\draw [style=causal consistency, bend left=15] (83) to (123);
		\draw [style=causal consistency b, bend right=15] (83) to (123);
		\draw [style=causal consistency b, bend right=15] (71) to (122);
		\draw [style=dotted line] (84) to (124.center);
		\draw [style=causal consistency] (83) to (122);
		\draw [style=causal consistency] (84) to (123);
		\draw [style=causal consistency] (78) to (121);
		\draw [style=causal consistency] (74) to (120);
		\draw [style=dotted line] (122) to (125.center);
		\draw [style=dotted line] (123) to (126.center);
		\draw [style=dotted line] (121) to (127.center);
		\draw [style=dotted line] (120) to (128.center);
	\end{pgfonlayer}
\end{tikzpicture}}}
    \caption{Depiction of the tilted consensus task, the configurations in neither dashed box are partially terminated. The system slice consists of all the infinitely many configurations in either of the dashed boxes. We exemplify here that system slices may necessarily be \demph{infinite}, as there is no finite system slice with a \demph{section}, as any configuration on the left has to decide $0$ and can thus not be in the connected to any configuration on the right. We formalize this in \cref{sec:sheaf-model}. This system frame grows upwards.}
    \label{fig:tilted-consensus}
  \end{figure}

  Clearly, $b$ just waits until something arrives, whereas $a$ can terminate immediately. As $b$ cannot distinguish whether it will end up in the $0$ or $1$ deciding half until it receives that message (i.e. $b$ doesn't know $a$'s value), no configuration $U$ where no message has arrived yet can be terminal. At the same time, every configuration $U$ is partially terminated as it is indistinguishable for $a$ from a terminal configuration (marked in the dashed terminal regions). The smallest system slice therefore has infinite size in this execution graph.
\end{example}

System slices naturally form a graph structure where we can define sheaves, where the set of global configurations are its vertices and the causal consistency relations its edges. This suffices for our purposes, but note that this consists of the more general structure of a \demph{cellular complex}, where vertices and edges are \(0\) and \(1\) cells, and the execution structure naturally satisfies its local finite requirement. Cellular complexes are fundamental objects in algebraic topology, and a full account can be found in \cite{hatcheralgtop}. Most importantly, a cellular complex \(X\) has an underlying partial ordering of its cells (vertices and edges), denoted \(P_X\), which we exemplify now.

Finally, in order to access the categorical machinery needed to define the task sheaves, we need need to introduce the concept of a \demph{cellular category} \cite{MakkaiRosicky14}, which categorifies the cellular complex \(X\). by viewing its associated poset \(P_X\) as a category, which preserves the topological structure while enabling a categorical perspective better suited for defining sheaves.

\begin{example}[Undirected Graph]\label{ex:graph}
  An undirected graph \(G = (V,E)\), with \(V\) a collection of vertices and \(E\) a collection of edges, gives rise to a cellular category \(\Cat{Cell}(G)\). Note that the data of an edge \(v_1 \xleftrightarrow{e} v_2\) consists of an unordered pair of vertices \(\{v_1, v_2\}\), and the incidence relation of an edge to a vertex satisfies the inclusion \(v_1 \xhookrightarrow{} \{v_1,v_2\} = e\). With this in mind, \(\Cat{Cell}(G)\) is obtained by constructing an object for each vertex \(v \in V\) and for each pair of vertices \(e \in E\), with an arrow \(v \to e\) whenever \(v \xhookrightarrow{} e\). Objects obtained from vertices are called \demph{\(0\)-cells} and those obtained from edges are called \demph{\(1\)-cells}. This construction preserves the information of the graph, while adding a categorical structure.
\end{example}

Observe that any system slice is an undirected graph and therefore induces a cellular category. We will make ample use of this in the the following.

\section{Task Sheaves} \label{sec:sheaf-model} 
In this section, we introduce the task sheaf as a mathematical framework for encoding the solvability constraints of a distributed task. Sheaves provide a formalism that captures both the global structure of a task definition and the local constraints which protocols must respect. We present two equivalent formulations. First, an explicit construction, used to build the intuition over the running examples. Then, we present its categorical foundations, which enables the use of cohomology later in \Cref{sec:cohomology}. We then state the main theorem relating the solvability of a task to the existence of sections in the appropriate task sheaf.

\subsection{The Task Sheaf} \label{sec:task-sheaf-explicit} 
Sheaves can be informally understood as a structure allowing to track data that is associated to pieces of a space. In our case, we will be tracking the task data, the possible solutions according to the task specification, across the possible global states that our distributed system may assume. As such, the runs of the system being analyzed will provide us with our base space: a \remph{system slice} \(\systemslice\), derived from its system frame \(\eframe\), turned into a cellular category \(\Cat{Cell}(\systemslice)\), as seen in \Cref{sec:system-slices}.

A sheaf defined on a cellular category is a cellular sheaf \footnote{See \cite{Curry14} for a thorough treatment of cellular sheaves, and \cite{MacLaneMoerdijk94} for an overview of sheaf theory.}. A key characteristic of \remph{cellular} sheaves is the combinatorial nature of this space, which allows us to look at discrete structures and their generalizations to higher dimensions, such as the graphs and cellular complexes. This nature also gives us access to a much simpler theory. In full generality, sheaves must be shown to respect a technical \remph{sheaf condition}, which is automatically satisfied in the case of cellular sheaves, as proven in \cite[Theorem 4.2.10]{Curry14}.

In a distributed system, configurations represent snapshots of the global state, while indistinguishability relations capture what each process can observe locally. The task specification defines which outputs are valid for given inputs. We introduce now the construction of a sheaf task, that connects those concepts and allows us to reason about task solvability.

\begin{definition}[Task Sheaf]\label{def:task-sheaf}
  Let \(\systemslice = (\gstates, \{\csim_{\proc}\}_{\proc \in \processes})\) be a system slice obtained of set of runs \(\runs\), and let \(\task = \langle \mathcal{I}, \mathcal{O}, \Delta \rangle\) be a task. The \demph{task sheaf} \(\F_{\systemslice,\task}\) is a cellular sheaf defined as follows.

    \begin{enumerate}
        \item (\demph{stalks of configurations}) For each configuration $\g \in \systemslice$, the stalk is $\F_{\systemslice,\task}(\g) = \Delta(\pi_I(\g))$, i.e., the set of possible valid output configurations given the input assignments in $\g$. \label{def:task-sheaf-vertex}

        \item (\demph{stalks of relations}) For each edge $(\g, \h)_{\proc} \in \csim_{\proc}$ between configurations $\g$ and $\h$ for process $\proc$, the stalk is $\F_{\systemslice,\task}((\g, \h)_{\proc}) = \{ \pi_{\proc}(x)\;|\;x \in \F_{\systemslice,\task}(\g) \cup \F_{\systemslice,\task}(\h) \}$, i.e., the set of possible values that process $\proc$ can choose in either of the adjacent configurations. \label{def:task-sheaf-edge}

        \item (\demph{restriction maps}) The restriction map from a configuration $\g$ to an edge $(\g, \h)_{\proc}$ is $\F_{\systemslice,\task, \g \trianglelefteq (\g,\h)_{\proc}} = \pi_{\proc}$, i.e., it projects an output configuration to its $\proc$'th entry. \label{def:task-sheaf-restriction}
    \end{enumerate}
\end{definition}

The vertices of $\F_{\systemslice,\task}$ are configurations, while the edges in $\csim_{\proc}$ connect configurations where $\proc$ should have the same decision value.

\begin{example}[Task sheaf after one iterations]\label{ex:sheaf-graph}
    In \Cref{fig:sheaf-graph} we depict a cellular category (note, that it is also the system frame of the lossy link synchronous message adversary) with objects representing the possible task solving outputs. 
    \begin{figure}[ht!]
        \makebox[\textwidth][c]{\scalebox{0.9}{\begin{tikzpicture}
	\begin{pgfonlayer}{nodelayer}
		\node [style=LLa] (19) at (13, 0) {$0,0\leftrightarrow$};
		\node [style=LLa] (20) at (9, 0) {$0,0\rightarrow$};
		\node [style=LLa] (21) at (17, 0) {$0,0\leftarrow$};
		\node [style=LLa] (22) at (5, 0) {$0,1\rightarrow$};
		\node [style=LLa] (23) at (1, 0) {$0,1\leftrightarrow$};
		\node [style=LLa] (24) at (-3, 0) {$0,1\leftarrow$};
		\node [style=LLa] (25) at (-3, -3) {$1,1\leftarrow$};
		\node [style=LLa] (26) at (1, -3) {$1,1\leftrightarrow$};
		\node [style=LLa] (27) at (5, -3) {$1,1\rightarrow$};
		\node [style=LLa] (28) at (9, -3) {$1,0\rightarrow$};
		\node [style=LLa] (29) at (13, -3) {$1,0\leftrightarrow$};
		\node [style=LLa] (30) at (17, -3) {$1,0\leftarrow$};
		\node [style=textnode] (63) at (13, -1.5) {$\{{{0}\choose{0}}, {{1}\choose{1}}\}$};
		\node [style=textnode] (64) at (13, 2) {$\{{{0}\choose{0}}\}$};
		\node [style=textnode] (65) at (1, -1.5) {$\{{{1}\choose{1}}\}$};
		\node [style=textnode] (66) at (1, 2) {$\{{{0}\choose{0}}, {{1}\choose{1}}\}$};
		\node [style=none] (69) at (19.5, -1.5) {};
		\node [style=none] (71) at (7, 0.5) {};
		\node [style=textnode] (75) at (7, 3.5) {$\pi_b( {{0}\choose{0}}) \cup \pi_b( {{1}\choose{1}} )$\\ $=\{0, 1\}$};
		\node [style=none] (78) at (19.5, -1.25) {};
		\node [style=none] (79) at (19.5, -1.75) {};
		\node [style=none] (80) at (20.5, 1) {$\pi_b(.)$};
		\node [style=none] (81) at (20.75, -3.25) {$\pi_b(.)$};
		\node [style=textnode] (82) at (24, -1) {$\pi_b( {{0}\choose{0}}) \cup \pi_b( {{1}\choose{1}} )$\\ $=\{0, 1\}$};
		\node [style=none] (83) at (6.75, 0.5) {};
		\node [style=none] (84) at (7.25, 0.5) {};
		\node [style=none] (85) at (8.75, 1.5) {$\pi_a(.)$};
		\node [style=none] (86) at (5.5, 1.5) {$\pi_a(.)$};
	\end{pgfonlayer}
	\begin{pgfonlayer}{edgelayer}
		\draw [style=b line] (19) to (20);
		\draw [style=b line] (22) to (23);
		\draw [style=b line, bend left=90, looseness=1.25] (25) to (24);
		\draw [style=b line] (26) to (27);
		\draw [style=b line] (28) to (29);
		\draw [style=b line, bend right=90, looseness=1.25] (30) to (21);
		\draw [style=a line] (21) to (19);
		\draw [style=a line] (20) to (22);
		\draw [style=a line] (23) to (24);
		\draw [style=a line] (25) to (26);
		\draw [style=a line] (27) to (28);
		\draw [style=a line] (29) to (30);
		\draw [style=dashed undirected] (20) to (64);
		\draw [style=dashed undirected] (19) to (64);
		\draw [style=dashed undirected] (21) to (64);
		\draw [style=dashed undirected] (28) to (63);
		\draw [style=dashed undirected] (29) to (63);
		\draw [style=dashed undirected] (63) to (30);
		\draw [style=dashed undirected] (27) to (65);
		\draw [style=dashed undirected] (26) to (65);
		\draw [style=dashed undirected] (25) to (65);
		\draw [style=dashed undirected] (24) to (66);
		\draw [style=dashed undirected] (23) to (66);
		\draw [style=dashed undirected] (22) to (66);
		\draw [style=dashed undirected] (71.center) to (75);
		\draw [style=restriction, bend left=90, looseness=2.00] (21) to (78.center);
		\draw [style=restriction, bend right=90, looseness=2.00] (30) to (79.center);
		\draw [style=dashed undirected] (69.center) to (82);
		\draw [style=restriction, bend left=45] (22) to (83.center);
		\draw [style=restriction, bend right=45] (20) to (84.center);
	\end{pgfonlayer}
\end{tikzpicture}}}
        \caption{Depiction of the system frame consisting of all configurations after one step of the lossy-link synchronous message adversary annotated with the binary consensus task stalks.
        For example when both inputs are $0$, then the stalk is a singleton containing just the decision vector ${0}\choose{0}$.}
        \label{fig:sheaf-graph}
    \end{figure}
    Indistinguishability edges are also assigned objects (in our case simply the union of both adjacent vertices' objects, i.e., all possible task outputs of the specific agent that cannot distinguish both configurations) and are omitted for space except at two edges. The restriction maps go from vertices to edges and map the global task output to the non-distinguishing agents output.
\end{example}

This explicit construction captures the constraints imposed by the task on the decisions of the processes. In order to establish its mathematical properties, we provide now its equivalent formulation as a colimit of task sheaves defines for each process.

\subsection{Categorical foundations} \label{sec:task-sheaf-categorical} 
The task sheaf has an equivalent process-wise formulation using the chromatic semi-simplicial sets (csets) as data modeling the distributed tasks. Csets were introduced by Goubault et al. \cite{GoubaultKniazevLedentRajsbaum23a}, and can be understood as sets of output values with enough structure to represent labeled configurations, a categorical generalization of the simplicial complexes used in \Cref{def:task}.

For each process \(\proc \in \processes\), a sheaf \(\F_{\systemslice,\task} : \Cat{Cell}(\systemslice) \To \Cat{CSet}_{\processes}\) is defined.
Whenever the system slice $\systemslice$ and the task $\task$ are clear from the context, we will write $\F$ and \(\F_\proc\). \(\F_\proc\) has the following structure:
\begin{itemize}
    \item \(\Cat{Cell}(\systemslice)\) is the cellular category obtained by localizing the system slice \(\systemslice\) to process \(\proc\), i.e., from the poset \(P_{\proc}\) induced by  \(\csim_{\proc}\) over \(\gstates\). Note that a pair of configurations \((\g, \h) \in \; \csim_{\proc}\) induces the relations \(\g \trianglelefteq (\g, \h)\) and \(\h \trianglelefteq (\g, \h)\).
    \item \(\Cat{CSet}_{\processes}\) is the category of chromatic semi-simplicial sets.
\end{itemize}

In this formalism, configurations \(\g\) become \(0\)-cells and equivalence causal consistency become \(1\)-cells. The sheaf \(\F_{\proc}\) maps a set of configurations to the set of \(n\)-simplices \(\Delta_{n}\) corresponding to its acceptable outputs, according to the task specification \(\task\), and maps an inclusion of configurations \(\g \xhookrightarrow{} (\g,\h)\) to face maps \(\F_{\g \trianglelefteq (\g,\h)} : \Delta_{n} \to \Delta_{1}\), where a set of output decisions is sent to decisions process-wise, i.e. its \(\proc\)th colored faces, such that the following diagram commutes.

\ctikzfig{sheaf-task}

The sheaves defined for all agents \(\{\F_{\proc}\}_{p_{i} \in \Pi}\) are a set of functors with common codomain.

\ctikzfig{sheaf-fibration}

These form a subcategory of the category of cellular sheaves, where each object is a cellular sheaf \(\F_\proc : \Cat{C}_\proc \to \Cat{D}\) and morphisms are commutative squares

\ctikzfig{sheaf-morphism}

where only all sheaves have a common codomain. The colimit of such cellular categories is well behaved and can be lifted for a colimit of cellular sheaves under those assumptions. The colimit of \(\{\F_{\proc}\}_{\proc \in \processes}\) then gives us a sheaf \(\F\) that captures all of the information on the individual sheaves, where a global section exists iff there is a global section on the individual ones. The sheaf \(\F\) is defined over each configuration \(\g\) (resp. indistinguishability edge \((\g,\h)\)) as the colimit of the individual sheaves localized at \(\g\) (resp. \((\g, \h)\)). The sheaf \(\F\) coincides with the explicit definition given in \Cref{def:task-sheaf}.

The colimit of \(\{\F_\proc\}_{\proc\in\processes}\) coincides with the task sheaf \(\F\) from \Cref{def:task-sheaf}. This categorical perspective ensures our construction is well-founded and enables the cohomology computations in \Cref{sec:cohomology}.

\subsection{Solvability as Sections} \label{sec:task-sheaf-solvability} 

We now define sections, which capture globally consistent assignments of data across the cellular complex. They are a fundamental concept in sheaf theory and will complete the language needed for our analysis of distributed task solvability.

\begin{definition}[Section] \label{def:sections}
  Let \(\F: \Cat{Cell}(P_X) \to \Cat{D}\) be a cellular sheaf over \((X, P_X)\). A \demph{section} of \(\F\) is a choice of elements \(s = \{s_\alpha \in \F(\alpha) \mid \alpha \in \Cat{Cell}(P_X)\}\) such that for every pair of cells with \(\alpha \trianglelefteq \beta\) the values coincide through restriction maps, i.e., \(\F_{\alpha \trianglelefteq \beta}(s_\beta) = s_\alpha\). The set of all global sections of \(\F\) is denoted by \(\Gamma(X, \F)\).
\end{definition}

A section is a choice of values, one per cell, such that the same value is obtained if we restrict to a cell from each of its incident neighbors. In the case of a task sheaf, it is a choice of data (i.e, output value) for each vertex (i.e, configuration) and each edge (i.e, causal consistency relation) such that they all agree under the restriction maps. Explicitly, it is an element of the direct sum of the stalks:
\[
  \Gamma(\F;\systemslice) =
  \bigl( \bigoplus_{\gstate \in \gstates} \F_{\systemslice, \task}(\gstate) \bigr) \times
  \bigl( \bigoplus_{(\gstate,\gstate') \in \{\csim\}_\{\proc\in\processes} \F_{\systemslice, \task}((\gstate,\gstate')) \bigr) .
\] 

We can now state our main characterization theorem, which establishes that task solvability is equivalent to the existence of sections in our task sheaf.

\begin{theorem}[Terminating Task Solvability]\label{thm:term_task_solvable}
    Let $T = \langle \mathcal{I}, \mathcal{O}, \Delta \rangle$ be a task and $\runs$ a system: there exists a \demph{terminating decision map} $\delta$ solving $T$ iff there exists a \demph{execution cut} $A$ such that its system slice, $\ccl(\slice(A))$, together with it causal consistency relation $\{ \simeq_{\proc} \}_{\proc\in\processes}$, has a section over the task sheaf $\F_{A, \simeq}$.
\end{theorem}

\begin{proof}
    Let us assume first that $\delta$, solves $T$. Therefore, for any run $\run$, there exists a configuration $\g_\run$ that is the earliest configuration of $\run$ where each process has decided $\delta(\g) \neq \bot$. Note that $A = \{ \g_\run \: \vert \: \run \in \Sigma \}$ is a system cut.

    Consider the system slice $\ccl(\slice(A))$. By assumption $\delta_p$ assigns a $T$ solving, non-$\bot$ value to every configuration $\g \in \ccl(\slice(A))$, and satisfies \cref{def:task-sheaf-vertex}, \cref{def:task-sheaf-restriction} and \cref{def:task-sheaf-edge} making $\F_{A,\simeq}$ a sheaf. The section condition in \cref{def:sections} is satisfied as decisions are final, i.e., any successor configuration of a terminated configuration has the same decision, and $\delta_p$ is a function on local states, i.e., configurations $p$ cannot distinguish are mapped identically by $\delta_p$.

    Now to prove the converse, assume that there exists an execution cut $A$ such that its system slice, $\ccl(\slice(A))$, has a section $S$. We construct a terminating decision map by setting
    $$ \delta_p(\pi_{p}(\g)) = \begin{cases}
            S(\phi(A,\h)) &\text{ for any $\h$ where } \pi_p(\h) = \pi_p(\g) \text{ and } \h \in A\\
            \bot  &\text{ else.}
        \end{cases} $$
    Clearly $\delta_p$ terminates because $A$ is an execution cut and satisfies \cref{def:task-validity}.
    $\delta_p$ solves $T$ because first, the only vectors in the stalks over the configurations are task solving vectors. And second, $\delta_p$ is well-behaved and maps identical inputs to identical outputs by \Cref{def:task-sheaf-restriction} and the section condition in \Cref{def:sections}, i.e, per-processes outputs need to agree over indistinguishable configurations.
\end{proof}

See the following \cref{ex:term-task-sheaf} for an application of \cref{thm:term_task_solvable}.

\begin{example}[Consensus on Task Sheaf]\label{ex:term-task-sheaf}
    In \cref{fig:sheaf_example} we depict a section over a task sheaf defined on a system slice. The vectors over the configurations in the system slice correspond to the valid choice of data that the sheaf took. Restriction maps and stalks are not depicted.

    \begin{figure}[ht!]
        \centering
        \makebox[\textwidth][c]{\scalebox{0.8}{\begin{tikzpicture}
	\begin{pgfonlayer}{nodelayer}
		\node [style=LLa] (5) at (3, -7) {$2,1$};
		\node [style=LLa] (17) at (-2.75, -7) {$0,1$};
		\node [style=LLa] (72) at (-2.75, -4) {$0,1\leftarrow$};
		\node [style=partially terminated] (73) at (3, -4) {$2,1\leftarrow$};
		\node [style=LLa] (74) at (-8, -4) {$0,1\leftrightarrow$};
		\node [style=partially terminated] (83) at (8, -4) {$2,1\leftrightarrow$};
		\node [style=LLa] (120) at (-13, -4) {$0,1\rightarrow$};
		\node [style=LLa] (122) at (16, -4) {$2,1\rightarrow$};
		\node [style=LLa] (123) at (3.5, -1) {$2,1\leftarrow\rightarrow$};
		\node [style=LLa] (124) at (7.5, -1) {$2,1\leftrightarrow\rightarrow$};
		\node [style=LLa] (125) at (-0.5, -1) {$2,1\leftarrow\leftrightarrow$};
		\node [style=LLa] (126) at (11.5, -1) {$2,1\leftrightarrow\leftrightarrow$};
		\node [style=LLa] (127) at (-4.5, -1) {$2,1\leftarrow\leftarrow$};
		\node [style=LLa] (128) at (15.5, -1) {$2,1\leftrightarrow\leftarrow$};
		\node [style=none] (129) at (17.5, 0) {};
		\node [style=none] (130) at (18, -0.5) {};
		\node [style=none] (131) at (18, -4.5) {};
		\node [style=none] (132) at (17.5, -5) {};
		\node [style=none] (133) at (-0.5, -2) {};
		\node [style=none] (134) at (-1, -2.5) {};
		\node [style=none] (135) at (-1, -4.5) {};
		\node [style=none] (136) at (-1.5, -5) {};
		\node [style=none] (137) at (-14.25, -5) {};
		\node [style=none] (138) at (-14.75, -4.5) {};
		\node [style=none] (139) at (-14.75, -3.5) {};
		\node [style=none] (140) at (-14.25, -3) {};
		\node [style=none] (141) at (-6, 0) {};
		\node [style=none] (142) at (-6.5, -0.5) {};
		\node [style=none] (143) at (-6.5, -2.5) {};
		\node [style=none] (144) at (-7, -3) {};
		\node [style=none] (145) at (14.5, -5) {};
		\node [style=none] (146) at (14, -4.5) {};
		\node [style=none] (147) at (14, -2.5) {};
		\node [style=none] (148) at (13.5, -2) {};
		\node [style=none] (149) at (-0.75, -1) {};
		\node [style=none] (150) at (-0.75, -1) {};
		\node [style=stalk] (151) at (-15, -2) {$0\choose{0}$};
		\node [style=stalk] (152) at (-10, -2) {$0.5\choose{0}$};
		\node [style=stalk] (153) at (-6.25, -7) {$0.5\choose{1}$};
		\node [style=stalk] (154) at (-3.25, 1) {$1.5\choose{1}$};
		\node [style=stalk] (155) at (0.75, 1) {$1.5\choose{1}$};
		\node [style=stalk] (156) at (4.75, 1) {$1.5\choose{1}$};
		\node [style=stalk] (157) at (8.75, 1) {$1.5\choose{2}$};
		\node [style=stalk] (158) at (12.75, 1) {$1.5\choose{2}$};
		\node [style=stalk] (159) at (16.75, 1) {$1.5\choose{2}$};
		\node [style=stalk] (160) at (17, -7) {$2\choose{2}$};
		\node [style=none] (171) at (-16.5, -1) {2:};
		\node [style=none] (172) at (-16.5, -4) {1:};
		\node [style=none] (173) at (-16.5, -7) {0:};
	\end{pgfonlayer}
	\begin{pgfonlayer}{edgelayer}
		\draw [style=b line] (17) to (5);
		\draw [style=b line] (73) to (72);
		\draw [style=a line] (74) to (72);
		\draw [style=a line] (73) to (83);
		\draw [style=b line] (83) to (122);
		\draw [style=b line] (74) to (120);
		\draw [style=dashed arrow] (5) to (122);
		\draw [style=dashed arrow] (5) to (73);
		\draw [style=dashed arrow] (5) to (83);
		\draw [style=dashed arrow] (17) to (72);
		\draw [style=dashed arrow] (17) to (74);
		\draw [style=dashed arrow] (17) to (120);
		\draw [style=dashed arrow] (73) to (125);
		\draw [style=dashed arrow] (73) to (123);
		\draw [style=dashed arrow] (73) to (127);
		\draw [style=dashed arrow] (83) to (124);
		\draw [style=dashed arrow] (83) to (126);
		\draw [style=dashed arrow] (83) to (128);
		\draw [style=a line] (123) to (124);
		\draw [style=a line] (125) to (127);
		\draw [style=b line] (124) to (126);
		\draw [style=b line] (125) to (123);
		\draw [style=a line] (126) to (128);
		\draw [style=dashed undirected, bend right=45] (140.center) to (139.center);
		\draw [style=dashed undirected] (139.center) to (138.center);
		\draw [style=dashed undirected, bend right=45] (138.center) to (137.center);
		\draw [style=dashed undirected] (137.center) to (136.center);
		\draw [style=dashed undirected, bend right=45] (136.center) to (135.center);
		\draw [style=dashed undirected] (135.center) to (134.center);
		\draw [style=dashed undirected, bend left=45] (134.center) to (133.center);
		\draw [style=dashed undirected, bend right=45] (132.center) to (131.center);
		\draw [style=dashed undirected] (131.center) to (130.center);
		\draw [style=dashed undirected, bend right=45, looseness=1.25] (130.center) to (129.center);
		\draw [style=causal consistency, bend left=15] (74) to (72);
		\draw [style=causal consistency b, bend right=15] (72) to (127);
		\draw [style=causal consistency b, bend left=15] (120) to (74);
		\draw [style=causal consistency b, bend left=15] (125) to (123);
		\draw [style=causal consistency b, bend left=15] (124) to (126);
		\draw [style=causal consistency, bend left=15] (126) to (128);
		\draw [style=causal consistency, bend left=345] (124) to (123);
		\draw [style=causal consistency, bend left=345] (125) to (127);
		\draw [style=causal consistency b, bend right=15] (72) to (73);
		\draw [style=causal consistency b, bend right=15, looseness=1.25] (73) to (123);
		\draw [style=causal consistency b] (72) to (125);
		\draw [style=dashed undirected, bend right=45] (144.center) to (143.center);
		\draw [style=dashed undirected] (143.center) to (142.center);
		\draw [style=dashed undirected, bend left=45] (142.center) to (141.center);
		\draw [style=dashed undirected] (140.center) to (144.center);
		\draw [style=dashed undirected] (141.center) to (129.center);
		\draw [style=dashed undirected] (132.center) to (145.center);
		\draw [style=dashed undirected, bend left=45] (145.center) to (146.center);
		\draw [style=dashed undirected] (146.center) to (147.center);
		\draw [style=dashed undirected, bend right=45] (147.center) to (148.center);
		\draw [style=dashed undirected] (133.center) to (148.center);
		\draw [style=causal consistency b, bend left=15] (122) to (128);
		\draw [style=causal consistency b] (122) to (126);
		\draw [style=causal consistency b, bend left=15] (83) to (124);
		\draw [style=causal consistency b, bend right=15] (83) to (122);
		\draw [style=causal consistency b, bend right=15] (73) to (125);
		\draw [style=causal consistency b, bend left=15] (83) to (126);
		\draw (151) to (120);
		\draw (152) to (74);
		\draw (153) to (72);
		\draw (154) to (127);
		\draw (155) to (125);
		\draw (156) to (123);
		\draw (157) to (124);
		\draw (158) to (126);
		\draw (159) to (128);
		\draw (122) to (160);
	\end{pgfonlayer}
\end{tikzpicture}}}
        \caption{We again consider the lossy link synchronous message adversary together with the epsilon agreement task. The configurations within the dashed polygon constitute a system slice solving the $\epsilon = 0.5$ agreement task in the case where the only allowed inputs are $(0,1)$ or $(1,2)$. The solid colored lines represent indistinguishabilities, the dashed lines represent the causal consistency relation. The system grows upwards.}
      \label{fig:sheaf_example}
    \end{figure}
\end{example}

Now that we have characterized solvability of a task $T$ under a protocol $\protocol$ in terms of its task sheaf, $\F_{A,\simeq}$, we will also establish a relation between local cohomology of the task sheaf and the task solvability.

\section{Computing Solutions} 
\label{sec:cohomology}

Given a communication adversary $K$ representing some adversarial entity together with a task $T$, we search for a protocol together with a decision map. As the full-information protocol provides the finest possible execution tree, its natural to start the search process there. Slices that are finite in size can be recursively enumerated and, although that's computationally inefficient, tested. If we find a finite slice $A$ that allows for a terminal decision map we can try to optimize the protocol inducing the execution tree. If we do not find a slice, but keep looking forever, then a finite slice does not exist and the task is not wait-free solvable, although it might be solvable given an \demph{infinite slice}.

In this section we focus on wait-free solvability, so assume we have found a finite terminating slice $\ccl(\slice(A))$, we computationally determine the space of all sections and therefore the decision map $\delta$. We do this via \emph{cohomology}, i.e., we turn our structure into an abelian group and extract from it all possible solutions.

Given a task sheaf \(\F\), the process of obtaining the \(n\)th sheaf cohomology can be understood as an iteration of the following steps.
\[
  H^{n}(\F) : \Cat{Cell}(X) \xrightarrow{\F} \Cat{CSet}_{\Pi} \xrightarrow{\Z} \Cat{sAb} \xrightarrow{\delta^{n}} \Cat{Ch}_{\Z} \xrightarrow{H^{n}} \Cat{Ab}
\]

Where \(\Cat{sAb}\) is the category of simplicial abelian groups, \(\Cat{Ch_{\Z}}\) is the category of (co)chain complexes with integer coefficients and \(\Cat{Ab}\) the category of Abelian groups, where our cohomology lives.

We provide a brief explanation, and illustrate it below in \Cref{ex:approx_agreement}. Given a cset \(\mathcal{O}\) containing the possible system output states, we can obtain a simplicial abelian group through the left adjoint of the forgetful functor that sends it to the underlying simplicial set\footnote{Here chromatic semi-simplicial sets are treated as simplicial sets for simplicity, as they only add the process labeling that would require extra bookkeeping.}. The third map gets us a cochain complex with differentials defined from the alternating sum of the (co)face maps. Finally, we obtain the \(n\)-th cohomology group, which corresponds to a space of sections of our sheaf. This construction is well known in the literature of algebraic topology \cite{hatcheralgtop} and we adapt it as a tool for understanding distributed tasks.

\begin{definition}[Space of Zero- and One- Cochains]
    Resembling \Cref{def:sections}, we define $$C^0(\eframe; \F) = \underset{v\in V(\eframe)}{\bigoplus} \F(v)$$ as the space of \demph{zero cochains} of the sheaf $\F$, i.e., the vector space of all possible assignments to vertices (configurations) in $\F$. Similarly $$C^1(\eframe; \F) = \underset{e_{\proc} \in E(\eframe)}{\bigoplus} \F(e_{\proc})$$ is the space one \demph{one cochains}, i.e., the space over possible output choices for processors.
\end{definition}

The two cochain groups are connected via a linear coboundary map $\delta$. This maps a specific choice of output vectors to the individual choices along the indistinguishability edges defined by the restriction maps. To define $\delta$ we chose an arbitrary direction on each indistinguishability edge $e = \g \rightarrow \h$ just to facilitate an algebraic representation.

\begin{definition}[Coboundary Map]
    We denote by $d: C^0 \mapsto C^1$ the coboundary map, defined per edge $e = (\g, \h)_{\proc}$ in the sheaf as:
    $$ d(x)_{e_{\proc}=(\g,\h)_{\proc}} = \pi_{\proc}(x_{\h}) - \pi_{\proc}(x_{\g}) , $$
    where we assume that the chosen direction goes from $\g \rightarrow \h$.
\end{definition}

We can represent the coboundary map $d$ as a coboundary matrix $D$ where rows are indexed by edges and columns are indexed by configurations:
$$D_{\g,e=(\h,\h')_{p_j}} = \begin{cases}
    d(e) & \text{ if } \h = \g \text{ or } \h' = \g \\
    0    & \text{ otherwise. }
\end{cases}$$
One can think of $D$ simply as computing the difference between two indistinguishable configurations. A section on our sheaf is a $0$-cochain that is mapped to $0$ by $d$, so any assignment to configurations such that any process that cannot distinguish two configurations, decides the same thing. The set of all sections is then the kernel $ker(D)$.

\begin{definition}[Zeroth Cohomology]
    The zero-th cohomology is $ker(D)$, i. e., the kernel of the coboundary map.
\end{definition}

\begin{example}\label{ex:approx_agreement}
    Let us consider the approximate agreement problem for 2 processes in the lossy link synchronous message adversary setting, illustrated in \Cref{fig:approx-LL-example}. We are interested in whether given a full information protocol, the induced system execution graph allows for a terminal execution cut, such that by \Cref{thm:term_task_solvable} we can find a section that gives us a decision map.

    We set the possible input vectors to $I = \{ {{0}\choose{0}},{{1}\choose{0}},{{0}\choose{1}},{{1}\choose{1}}\}$, the possible output vectors to $O = \{ {{x}\choose{y}}\;|\; x,y \in \{0,0.25,0.5,0.75,1\}\}$ and define the validity map as 
    $$\Delta({{x}\choose{y}}) = \begin{cases} {{0}\choose{0}} & \text{ when } x = y = 0 \\ {{1}\choose{1}} & \text{ when } x = y = 1 \\ O & \text{ otherwise.} \end{cases}$$
    Intuitively, in the initial configuration we cannot find a section, since both configurations, $(0,0)$ and $(1,1)$, force the respective connecting configurations $(0,1)$ and $(1,0)$ to choose an output vector that projects to $1$ and $0$, which does not exist in $O$. We can formalize this impossibility starting with the co-boundary matrix $D^0$, where we  number the configurations by the initial values interpreted in binary, and edges are just tuples of configurations, with the direction following the written order. Note that this is the co-boundary matrix of the system slice consisting of all \demph{initial configurations}, i.e., after $0$ steps. The co-boundary matrix for all configurations after one step has the unreadable dimensions of $12x12$.

    $$ D^0 = \begin{pmatrix}
           & 0 & 1 & 2 & 3 \\
        01 & -\pi_a(.) & \pi_a(.) & & \\
        13 & & -\pi_b(.) & & \pi_b(.) \\
        32 & & & \pi_a(.) & -\pi_a(.) \\
        20 & \pi_b(.) & & \pi_b(.) & \\
    \end{pmatrix}$$
    In order to find the kernel of $D^0$, we can assume some arbitrary assignment vector to the configurations $x = ({{0}\choose{0}},x_1,x_2,{{1}\choose{1}})^T$ (as the configurations $1$ and $3$ have only one possible choice by validity) and solve $D^0x=0$:
    $$ D^0x = \begin{pmatrix}
        -0^a + x^a_1 \\
        -x^b_1 + 1^b \\
        x^a_2 - 1^a \\
        0^b - x^b_2 \\
    \end{pmatrix} \implies \begin{matrix}
        x^a_1 = 0^a, \\
        x^b_1 = 1^b, \\
        x^a_2 = 1^a, \\
        x^b_2 = 0^b \\
    \end{matrix} \implies x_1 = \begin{pmatrix}
        0 \\
        1
    \end{pmatrix} \text{ and } x_2 = \begin{pmatrix}
        1 \\
        0
    \end{pmatrix}.$$
    This proves that epsilon agreement is impossible in $0$ steps as the required solutions $x_1$ and $x_2$ are not possible solutions, the kernel is trivial. The impossibility itself does not come as a surprise. The novelty lies in the fact that that every step we took was purely deterministic and computable, meaning that such operations could have been done by a program.

    We depict the system slice consisting of all configurations after \demph{one step} in \cref{fig:approx-LL-example}. $(1,1\rightarrow)$ forces $(1,0\rightarrow)$ to choose an output that matches $a$'s decision, i.e., ${{1}\choose{.75}}$. This forces $b$'s decision in $(1,0\leftrightarrow)$ to be $.75$ a valid output could be ${{.75}\choose{.5}}$. Again, this forces $a$'s hand in $(1,0\leftarrow)$ to $.5$, we could choose ${{.5}\choose{.25}}$ here. But now we run into trouble, in configuration $(0,0\leftarrow)$ we cannot find an output vector that maps $b$'s value to $.25$ as the only possible choice here is ${0}\choose{0}$! Therefore this assignment is not a section!

    \begin{figure}[ht!]
        \makebox[\textwidth][c]{\begin{tikzpicture}
	\begin{pgfonlayer}{nodelayer}
		\node [style=LLa] (19) at (13, 0) {$0,0\leftrightarrow$};
		\node [style=LLa] (20) at (9, 0) {$0,0\rightarrow$};
		\node [style=LLa] (21) at (17, 0) {$0,0\leftarrow$};
		\node [style=LLa] (22) at (5, 0) {$0,1\rightarrow$};
		\node [style=LLa] (23) at (1, 0) {$0,1\leftrightarrow$};
		\node [style=LLa] (24) at (-3, 0) {$0,1\leftarrow$};
		\node [style=LLa] (25) at (-3, -3) {$1,1\leftarrow$};
		\node [style=LLa] (26) at (1, -3) {$1,1\leftrightarrow$};
		\node [style=LLa] (27) at (5, -3) {$1,1\rightarrow$};
		\node [style=LLa] (28) at (9, -3) {$1,0\rightarrow$};
		\node [style=LLa] (29) at (13, -3) {$1,0\leftrightarrow$};
		\node [style=LLa] (30) at (17, -3) {$1,0\leftarrow$};
		\node [style=textnode] (64) at (13, 2) {$\{{{0}\choose{0}}\}$};
		\node [style=textnode] (65) at (1, -1.5) {$\{{{1}\choose{1}}\}$};
		\node [style=textnode] (66) at (1, 2) {$\mathcal{O}$};
		\node [style=partially terminated] (69) at (20, -1.5) {$0\neq.25$};
		\node [style=none] (85) at (8.75, -1.5) {$\pi_a{{1}\choose{.75}}$};
		\node [style=none] (86) at (5.25, -4.5) {$\pi_a{{1}\choose{1}}$};
		\node [style=none] (91) at (9.25, -4.5) {$\pi_b{{1}\choose{.75}}$};
		\node [style=none] (92) at (12.75, -1.5) {$\pi_b{{.5}\choose{.75}}$};
		\node [style=none] (94) at (13.25, -4.5) {$\pi_a{{.5}\choose{.75}}$};
		\node [style=none] (97) at (16.75, -1.5) {$\pi_a{{.5}\choose{.25}}$};
		\node [style=none] (98) at (21.5, -3.25) {$\pi_b{{.5}\choose{.25}}$};
		\node [style=none] (99) at (21.5, 0.25) {$\pi_b{{0}\choose{0}}$};
		\node [style=partially terminated] (100) at (15, -3) {$.5$};
		\node [style=partially terminated] (101) at (7, -3) {$1$};
		\node [style=partially terminated] (102) at (11, -3) {$.75$};
	\end{pgfonlayer}
	\begin{pgfonlayer}{edgelayer}
		\draw [style=b line] (19) to (20);
		\draw [style=b line] (22) to (23);
		\draw [style=b line, bend left=90, looseness=1.25] (25) to (24);
		\draw [style=b line] (26) to (27);
		\draw [style=b line] (28) to (29);
		\draw [style=b line, bend right=90, looseness=1.25] (30) to (21);
		\draw [style=a line] (21) to (19);
		\draw [style=a line] (20) to (22);
		\draw [style=a line] (23) to (24);
		\draw [style=a line] (25) to (26);
		\draw [style=a line] (27) to (28);
		\draw [style=a line] (29) to (30);
		\draw [style=dashed undirected] (20) to (64);
		\draw [style=dashed undirected] (19) to (64);
		\draw [style=dashed undirected] (21) to (64);
		\draw [style=dashed undirected] (27) to (65);
		\draw [style=dashed undirected] (26) to (65);
		\draw [style=dashed undirected] (25) to (65);
		\draw [style=dashed undirected] (24) to (66);
		\draw [style=dashed undirected] (23) to (66);
		\draw [style=dashed undirected] (22) to (66);
		\draw [style=restriction, bend left=300] (30) to (100);
		\draw [style=restriction, bend right=60] (29) to (100);
		\draw [style=restriction, bend right=60] (27) to (101);
		\draw [style=restriction, bend right=60] (28) to (101);
		\draw [style=restriction, bend right=60] (28) to (102);
		\draw [style=restriction, bend right=60] (29) to (102);
		\draw [style=restriction, bend left=90, looseness=1.25] (21) to (69);
		\draw [style=restriction, bend right=75, looseness=1.25] (30) to (69);
	\end{pgfonlayer}
\end{tikzpicture}}
        \caption{An example of the approximate agreement task (with $\epsilon = 0.25$) in the lossy-link synchronous communication adversary, as in \cref{ex:system-frame}. The system slice consists again of all configurations after \demph{one round}, we depict a possible assignment of task solving vectors to the configurations. Starting from $(1,1\rightarrow)$ we then draw all the restriction maps and the following derived vectors until we arrive at an assignment where the restriction maps do not agree. Essentially this just shows that, this assignment is not a section.}
        \label{fig:approx-LL-example}
    \end{figure}

    Note that this example is not a proof that one cannot solve approximate agreement after one step, and is only meant to illustrate the role of cohomology in determining task solvability, as the co-boundary matrix after one step is already huge. But, as already illustrated, any step here is deterministic and computable, therefore we can find a section after two steps, implying the existence of a protocol solving epsilon agreement.
\end{example}

We formalize this intuition in \Cref{thm:computable_delta}.

\begin{theorem}[Computable Decision Maps]\label{thm:computable_delta}
    Let $T$ be a task that can be solved in a \emph{finite} system slice (i. e. finitely many terminal configurations) in a given execution graph \exec, then its decision map is computable in finite time.
\end{theorem}

The idea is simple: compute the execution graph \(\exec\) layer by layer and check whether any system slice admits a non-trivial zeroth cohomology. 

\begin{proof}
    Assume the task $T$ can be solved in a finite system slice. We can iteratively build up the tree $\exec(k)$ \footnote{As any configuration has at most finitely directly causally dependent configurations, i. e., children in the execution graph, we can label any node by its depth. Then we iteratively build the tree up to depth $n$.}, up to some $k$. For each $k$, chose any possible  $A$ and corresponding system slice $\ccl(\slice(A)) \subseteq \ccl(\slice(A))$ and compute its zeroth cohomology. If it admits a non-trivial kernel on $\ccl(\slice(A))$, then derive a protocol as described in \cref{thm:term_task_solvable}.

    By assumption such a $\slice(A)$ exists. By building the tree iteratively, eventually this $A$ will be found and the iterations terminates.
\end{proof}

\section{Conclusions} \label{sec:conc} 
Our results, and in particular, our task sheaf construction, constitutes to the best of our knowledge, the first sheaf-theoretic characterization of general distributed computing tasks. Moreover, the generality of our model allows us to describe a wide range of systems that only need to satisfy minimal assumptions, namely, that the set of processes is finite, and that the communication is produced via messages.

By expressing tasks as a sheaf, we are able to incorporate cohomology theory as a powerful tool for distributed systems. For instance, the cohomology of a task sheaf is a group that represents the ``obstructions'' or ``limitations'' in the distributed system that prevent a specific task to be solved. Moreover, we show an impossibility result by simply looking at the cohomology group of its task sheaf. However, the cohomology of a task sheaf is not only restricted to determining impossibilities, but it may also be used for explicitly finding a protocol. Thus, sheaf-cohomology is shown to be a powerful, and promising tool for obtaining novel results and insights in distributed computing.

Finally, the rigorous categorical foundation of our approach provides a solid starting point for further research, such as incorporating failure models or exploring the complexity of protocol synthesis for different task definitions.

\bibliography{refs}

\end{document}